\documentclass[runningheads,a4paper,orivec]{llncs}

\usepackage{mathbbol}
\usepackage{amsmath}
\usepackage{amsfonts}
\usepackage{amssymb}
\usepackage{float}
  \floatstyle{boxed} %of float package, for boxing the figures
    \restylefloat{figure}

\usepackage{prooftree}
\usepackage{color}

\usepackage{cmll} %for linear logic connectives
\usepackage{xcolor} %FOR TIKZ
\usepackage{xkeyval} %FOR TIKZ
\usepackage{tikz} %for making reduction trees

\DeclareFontFamily{U}{cmllr}{}
\DeclareFontShape{U}{cmllr}{m}{n}{
	<4.25>cmllr5
	<5><6><7><8><8.5><9><9.5>gen*cmllr
    <10->cmllr10}{}
\DeclareFontShape{U}{cmllr}{bx}{n}{
    <2.97498><4.25>cmllbx5
	<5-9>gen*cmllbx
	<10->cmllbx10}{}
\DeclareFontFamily{U}{bbold}{}
\DeclareFontShape{U}{bbold}{m}{n}{
	<2.97498><4.16495><4.25>bbold5
	<5>bbold5
	<6>bbold6
	<7>bbold7
	<8><8.5>bbold8
	<9>bbold9
	<9-11>bbold10
	<11-14>bbold12
	<14->bbold12}{}
\newtheorem{myremark}[theorem]{Remark}
\usepackage[%
color=gray!20]{todonotes}
\newcommand{\halfogre}{\mathbf{\Delta^\star}}
\newcommand{\nat}{\mathbb{N}}
\newcommand{\mes}[1]{\vert #1\vert}
\newcommand{\msto}{\to^*}
\newcommand{\dom}{\mathrm{dom}}
\newcommand{\mcup}{\uplus}
\newcommand{\Mfin}[1]{\mathcal{M}_{\mathrm{f}}(#1)}
\newcommand{\FV}{\mathrm{FV}}
\newcommand{\Types}{\mathbb{T}}
\newcommand{\LTypes}{\mathbb{C}}
\newcommand{\Terms}{\Lambda_{+\parallel}}
\newcommand{\Values}{\mathrm{V_{+\parallel}}}
\newcommand{\Var}{\mathrm{Var}}
\newcommand{\lam}{\ensuremath{\lambda}}

\newcommand{\plusL}{\ensuremath{+_\ell}}
\newcommand{\plusR}{\ensuremath{+_r}}
\newcommand{\firstsort}[1][]{parallel-type{#1}}   %\firstsort{s} for mixed-types, \firstsort for singular
\newcommand{\secondsort}[1][]{computational-type{#1}} % analogous
\newcommand{\one}{\ensuremath{\textrm{\bf 1}}}
\newcommand{\Int}[1]{\mathbb{\Lbrack} #1\mathbb{\Rbrack}}
\newcommand{\sqle}{\sqsubseteq}
\newcommand{\ogre}{\mathrm{Y}^{\star}}

\newcommand{\keywords}[1]{\par\addvspace\baselineskip
\noindent\keywordname\enspace\ignorespaces#1}

\begin{document}
\mainmatter  % start of an individual contribution

% first the title is needed
\title{Call-by-Value Non-determinism\\ in a Linear Logic Type Discipline}
% a short form should be given in case it is too long for the running head
\titlerunning{Call-by-Value Non-determinism in a Linear Logic Type Discipline}

% the name(s) of the author(s) follow(s) next
%
% NB: Chinese authors should write their first names(s) in front of
% their surnames. This ensures that the names appear correctly in
% the running heads and the author index.
%
\author{Alejandro D\'iaz-Caro\inst{1,}\thanks{Partially supported by grants from DIGITEO and R\'egion \^Ile-de-France.}
  \and Giulio Manzonetto\inst{1,2}
  \and Michele Pagani\inst{1,2}}
\institute{Universit\'e Paris 13, Sorbonne Paris Cit\'e, LIPN, F-93430, Villetaneuse, France
  \and
  CNRS, UMR 7030, F-93430, Villetaneuse, France
  }
\authorrunning{D\'iaz-Caro, Manzonetto, Pagani}

\maketitle

\begin{abstract}
We consider the call-by-value $\lambda$-calculus extended with a may-convergent non-deterministic choice and
a must-convergent parallel composition.
Inspired by recent works on the relational semantics of linear logic and non-idempotent intersection types,
we endow this calculus with a type system based on the so-called Girard's second translation of intuitionistic logic into linear logic.
We prove that a term is typable if and only if it is converging, and that its typing tree
carries enough information to give a bound on the length of its lazy call-by-value reduction.
Moreover, when the typing tree is minimal, such a bound becomes the exact length of the reduction.

\keywords{$\lambda$-calculus, linear logic, non-determinism, call-by-value.}
\end{abstract}

% !TEX root = ../main.tex
%
\section{Introduction}\label{sec:intro}

%\textbf{Typing termination.}
The intersection type discipline provides logical characterisations of operational properties of $\lambda$-terms,
namely of various notions of termination, like head-, weak- and strong-normalisation (see \cite{CoppoDezani78,Salle80}, and \cite{Krivine90} as a reference).
The basic idea is to look at types as the set of terms having a given computational property ---
%the functional connective $\alpha\rightarrow\beta$ becomes the set of terms mapping terms in the type $\alpha$ into terms in the type $\beta$ and the intersection $\alpha\cap\beta$ represents those terms having both type $\alpha$ and type $\beta$.
 the type $\alpha\cap\beta$ being the set of those terms enjoying both properties $\alpha$ and $\beta$. With this intuition in mind, the intersection is naturally idempotent ($\alpha\cap\alpha=\alpha$).
 
Another way to understand the intersection type discipline is as a deductive system for presenting the
compact elements of a specific reflexive Scott domain (see e.g.~\cite[\S 3.3]{AmadioC98}). The set %\footnote{This set turns out to be a filter (w.r.t. the $\subseteq$ pre-order and the meet $\cap$), so the models induced by intersection types systems are called \emph{filter models}.}
of types assigned to a closed term captures the interpretation of such a term in the associated domain. %with the intersection type system (see e.g.~\cite{BarendregtCD83}). 
Intersection types are then a powerful tool for enlightening the relations between denotational semantics, syntactical types and computational properties of programs.

Intersection types have been recently revisited in the setting of the relational semantics {\bf Rel} of Linear Logic (LL).  %\cite{BucciarelliEM07}, 
{\bf Rel} is a semantics providing a more \emph{quantitative} interpretation of the \lam-calculus than Scott domains. Loosely speaking, the relational interpretation of a $\lambda$-term $M$ not only tells us whether $M$ converges on an argument, but in case it does, it also provides information on the number of times $M$ needs to call\footnote{The notion of \emph{calling an argument} should be made precise by specifying an operational semantics, which is usually achieved through an evaluating machine.} its argument to converge.
Just like the intersection type discipline captures Scott domains, \emph{non-idempotent} intersection type systems 
represent relational models.
In this framework the type $\alpha_1\cap\cdots\cap\alpha_k$ may be more accurately represented as the finite multiset $[\alpha_1,\dots,\alpha_k]$.
The lack of idempotency is the key ingredient to model the resource sensitiveness of $\mathbf{Rel}$~---
while in the usual systems $M : \alpha\cap \beta$ stands for ``$M$ can be used either as data of type $\alpha$ or as data of type $\beta$'',
when the intersection is not idempotent the meaning of $M : [\alpha,\beta]$ becomes ``$M$ will be called \emph{once} as data of type $\alpha$ and \emph{once} as data of type $\beta$''.
Hence, types should no longer be understood as \emph{sets of terms}, but rather as \emph{sets of calls}
 to terms. %The logical presentation of the relational model $\mathcal{D}$ has been given in \cite{PaganiR10} and used to study the may and must convergence of Tranquilli's resource $\lambda$-calculus \cite{Tranquilli11}.

The first intersection type system based on $\mathbf{Rel}$ has been presented in~\cite{deCarvalho}, where de~Carvalho introduced system~R, a type discipline capturing the relational version of Engeler's model. More precisely, he proved that system~R, beyond characterising converging terms, carries information on the evaluation sequence as well ---
the size of a derivation tree typing a term is a bound on the number of steps needed to reach a normal form.
Similar results are obtained in~\cite{BernadetL11} for a variant of system~R characterising strong normalisation and giving a bound to the longest $\beta$-reduction sequence. 
More recently, Ehrhard introduced a non-idempotent intersection type system characterising the convergence in the call-by-value \lam-calculus~\cite{Ehrhard12}.
Also in this case, the size of a derivation tree bounds the length of the lazy (i.e.\ no evaluation under $\lambda$'s) call-by-value $\beta$-reduction sequence. % reducing that term to a value.
Our goal is to extend Ehrhard's system with non-determinism. 

Our starting point is \cite{BucciarelliEM12}, where it is shown that the relational model $\mathcal D$ of the call-by-name \lam-calculus provides a natural interpretation of both may and must non-determinism. Since {\bf Rel} interprets \lam-terms as relations, the
\emph{may}-convergent non-deterministic choice %($M+N$ converges iff $M$ converges \emph{or} $N$ converges) 
can be expressed in the model as the set-theoretical union. The \emph{must}-convergent parallel composition, instead, is interpreted by using the operation $\mathcal{D} \otimes \mathcal{D}\multimap\mathcal{D}$ obtained by combining the mix rule $\mathcal{D}\otimes \mathcal{D}\multimap \mathcal{D}\parr\mathcal{D}$ with the contraction rule $\mathcal{D}\parr\mathcal{D}\multimap \mathcal{D}$, this latter holding since the call-by-name model $\mathcal{D}$ has shape $\wn A$ for $A = \mathcal{D}^\nat\multimap\perp$.
%Moreover, $\mathcal{D}$ is suitable for modeling the \emph{must}-convergent parallel composition as well,
%by virtue of its logical nature and thanks to the presence of the \emph{mix rule} in $\mathbf{Rel}$,
%that allows to ``put together'' any two proofs whatsoever.
%Indeed $\mathcal{D}$ has shape $\wn A$ for $A = \mathcal{D}^\nat\multimap\perp$, so parallel composition is interpreted using the operation
%$\mathcal{D} \otimes \mathcal{D}\multimap\mathcal{D}$ obtained by combining the mix rule $\wn A\otimes \wn A\multimap \wn A\parr \wn A$
%with the contraction rule $\wn A\parr\wn A\multimap \wn A$. 
 We will show that the same principle (\emph{may}-convergence as \emph{union} of interpretations and \emph{must}-convergence as \emph{mix} rule plus contraction) still works in the call-by-value setting.%, extending Ehrhard's call-by-value types to a non-deterministic framework.

Ehrhard's call-by-value type system is based on the so-called ``second Girard's translation'' of intuitionistic logic into LL \cite{Girard87,MaraistOTW99}. The translation of a type $\alpha$ is actually given by two mutually defined mappings ($\alpha\mapsto \alpha^v$ and $\alpha\mapsto \alpha^c$) reflecting the two sorts (\emph{values} and \emph{computations}) at the basis of the call-by-value $\lambda$-calculus:
\begin{align*}
\iota^v&=\iota,&(\alpha\rightarrow\beta)^v&=\alpha^c\multimap\beta^c,&\alpha^c&=\oc\alpha^v,
\end{align*}
where $\iota$ is an atom. Hence, the relational model described by Ehrhard's typing system yields a solution to the equation $\mathcal{V}\simeq\oc \mathcal{V} \multimap\oc \mathcal{V}$ in $\mathbf{Rel}$. % (or to the equation $T\simeq\oc(T \multimap T)$ if we focus on terms instead that on values)\todom{added}.
Since in this semantics $\multimap$ is interpreted by the cartesian product and $\oc$ by finite multisets,
a functional type for a value in this system is a pair $(p,q)$ of types for computations, and a type for a computation is a multiset $[\alpha_1,\dots,\alpha_n]$ of value types
(representing $n$ calls to a single value that must behave as $\alpha_1, \dots, \alpha_n$).
%Here the intuition is that a computation of type $[\alpha_1,\dots,\alpha_n]$ represents $n$ calls to a single value that must behave respectively as $\alpha_1$, \dots, $\alpha_n$.

In order to deal with the must non-determinism, namely the parallel composition, we must add to the translation considered by Ehrhard a further exponential level, called here the \emph{parallel sort}:
%to Ehrhard's system:
\begin{align}\label{eq:intro_types}
\iota^v&=\iota,&(\alpha\rightarrow\beta)^v&=\alpha^c\multimap\beta^{\parallel},&\alpha^c&=\oc\alpha^v,&\alpha^{\parallel}&=\wn\alpha^c.
\end{align}
This translation enjoys the nice property of mapping the call-by-value $\lambda$-calculus into the polarised fragment of LL, as described by Laurent in \cite{LaurentPhD}. 
Then, our typing system is describing an object in $\mathbf{Rel}$ satisfying the equation $\mathcal{V}\simeq{\oc \mathcal{V} \multimap\wn\oc \mathcal{V}}$, where the $\wn$ connective is interpreted by the finite multiset operator. 
In this setting a value type is a pair $(p,[q_1,\dots,q_n])$ of a computational type $p$ and a parallel type, that is a multiset of computations $q_1,\dots, q_n$. 
Intuitively, a value of that type needs a computation of type $p$ to create a parallel composition of $n$ computations of types $q_1, \dots, q_n$, respectively.  
Notice that, following \cite{BucciarelliEM12}, the composition of the mix rule and the contraction one yields an operation 
$\wn\oc \mathcal{V}\otimes\wn\oc \mathcal{V}\multimap\wn\oc \mathcal{V}$  which is used to interpret the parallel composition.

To avoid a clumsy notation with multisets of multisets, we prefer to denote a $\oc$-multiset $[\alpha_1,\dots, \alpha_m]$ (the type of a computation) 
with the linear logic multiplicative conjunction $\alpha_1\otimes\dots\otimes\alpha_m$, 
a $\wn$-multiset $[q_1,\dots,q_n]$ (the type of a parallel composition of computations)  with the multiplicative disjunction $q_1\parr\dots\parr q_n$, 
and finally a pair $(p,[q_1,\dots,q_n])$ with the linear implication $p\multimap(q_1\parr\dots\parr q_n)$. 
Such a notation stresses the fact that the non-idempotent intersection type systems issued from $\mathbf{Rel}$ are essentially contained in the multiplicative fragment of 
LL (modulo the associativity, commutativity and neutrality equivalences).

\medskip
{\bf Contents.}
Several non-deterministic extensions of the \lam-calculus have been proposed in the literature, 
both in the call-by-name (e.g.\ \cite{BucciarelliEM12,DezanidP96}) and in the call-by-value setting (e.g.\ \cite{BoudolIC94,DezaniLP98}). 
In the present paper we focus on the call-by-value \lam-calculus, first introduced in \cite{PlotkinTCS75}, endowed with two binary operators $+$ and $\parallel$
representing non-deterministic choice and parallel composition, respectively. 
The resulting calculus, denoted here $\Terms$, is quite standard and its operational semantics is given in Section~\ref{sec:cbv-nd-machine}
through a machine performing lazy call-by-value reduction.
Following \cite{BucciarelliEM12}, we model non-deterministic choice as \emph{may} non-determinism and 
parallel composition as \emph{must}. % non-determinism.
This is reflected in our reduction and in our notion of convergence.
Indeed, every time the machine encounters $M + N$ in active position it actually performs a choice, while encountering $M\parallel N$ 
it interleaves reductions in $M$ and in $N$; finally a term $M$ converges when there is a reduction of the machine from $M$ to a normal form.

Section~\ref{sect:types} is devoted to provide the type discipline for $\Terms$, based on the multiplicative fragment of LL (as discussed above),
and to define a measure $\mes{\cdot}$ associating a number with every type derivation.
Such a measure ``extracts'' from the information present in the typing tree of a term, a bound on the length of its evaluation.
In Section~\ref{sec:properties} we show that our type system satisfies good properties like subject reduction and expansion.
We also prove that the measure associated with the typing tree of a term decreases by 1 at every reduction step,
giving thus a proof of weak normalisation in $\omega$ for typable terms.
From these properties it ensues directly that a term is typable if and only if it converges.
Moreover, thanks to the resource consciousness of our type system, we are able to strengthen such a result ---
we prove that, whenever $M$ converges, there is a type derivation $\vdash M : \alpha$ (with $\alpha$ satisfying a suitable minimality condition) such that the associated measure provides the exact number of steps reducing $M$ to a normal form.

Finally, in Section~\ref{sec:NoFA} we discuss the properties of the model in $\mathbf{Rel}$ underlying our system.
As expected, the interpretation turns out to be adequate, i.e.\ a term converges if and only if its interpretation is non-empty.
On the other hand such a model is not fully abstract --- there are terms having different interpretations and that cannot be (semi-)separated
using applicative contexts. 
Our counterexample does not rely on the presence of  $+$ and $\parallel$~.

% !TEX root = ../main.tex
%
\section{The call-by-value non-deterministic machine}\label{sec:cbv-nd-machine}

We consider the call-by-value $\lambda$-calculus \cite{PlotkinTCS75}, extended with non-deterministic and parallel operators in the spirit of \cite{BucciarelliEM12}.
The set $\Terms$ of \emph{terms} and the set $\Values$ of \emph{values} are defined by mutual induction as follows
(where $x$ ranges over a countable set $\Var$ of variables):
$$
\begin{array}[t]{l@{\hspace{1.5cm}}r@{\ ::=\quad}l@{\hspace{1.5cm}}l}
      \textrm{\itshape Terms:} & M,N,P,Q & V \mid  MN \mid M+N \mid M\parallel N&\Terms\\
      \textrm{\itshape Values:} & V & x\mid \lambda x.M&\Values\\
\end{array}
$$
Intuitively, $M + N$ denotes the \emph{non-deterministic choice} between $M$ and $N$, while $M\parallel N$ stands for their parallel composition. 
Such operators are not required to be associative nor commutative. %\todom{added for rev4. do you agree?}\todoa{Agreed. In fact, looking at the red. rules, they are not (or there is redundance)}
As usual, we suppose that application associates to the left and \lam-abstraction to the right.
Moreover, to lighten the notation, we assume that application and \lam-abstraction take precedence over $+$ and $\parallel\ $.

%Concerning specific terms, we set:
%$$
%	\Omega = (\lam x.xx)(\lam x.xx)\qquad \mathbf{I} = \lam x.x
%$$

The \emph{$\alpha$-conversion} and the set $\FV(M)$ of \emph{free variables of $M$} are defined as usual in \lam-calculus \cite[\S2.1]{Bare}.
A term $M$ is \emph{closed} whenever $\FV(M) = \emptyset$.

Given $M\in\Terms$ and $V\in\Values$, we denote by $M[V/x]$ the term obtained by simultaneously  substituting the value $V$ for all free occurrences of $x$ in $M$, subject to the usual proviso about renaming bound variables in $M$ to avoid capture of free variables in $V$.
Hereafter terms are considered up to $\alpha$-conversion.

\begin{figure}[t]
  \begin{tabular}{p{4.8cm}p{3.5cm}p{4.3cm}}
  \emph{$\beta_v$-reduction} &  \emph{$+$-reductions} & \emph{$\parallel$-reductions}\\[0.2ex]
  $(\lambda x.M)V\to M[V/x]$
  &
  $M+N\to M$

  $M+N\to N$
  &
  $(M\parallel N)P\to MP\parallel NP$

  $V(M\parallel N)\to VM\parallel VN$
  \end{tabular}

  ~{\small \emph{Contextual rules}\vspace{-0.5ex}
  $$\begin{array}{c@{\hspace{0.4cm}}c@{\hspace{0.4cm}}c@{\hspace{0.4cm}}c}
    \prooftree M\to M'
    \justifies M\parallel N\to M'\parallel N
    \endprooftree
	&
    \prooftree N\to N'
    \justifies M\parallel N\to M\parallel N'
    \endprooftree
	&
    \prooftree M\to M'\hspace{0.3cm}(*)%M\neq P\parallel Q
    \justifies MN\to M'N

    \endprooftree
	&
    \prooftree M\to M'\hspace{0.3cm}(*)%M\neq P\parallel Q
    \justifies VM\to VM'
    \endprooftree
    \end{array}$$
    }
    \vspace{-1ex}
  \caption{Reduction semantics for $\Terms$. The condition $(*)$ stands for ``$M\neq P\parallel Q$''.}
  \label{fig:Machine}
\end{figure}

\begin{definition}[Operational semantics]
The operational semantics of $\Terms$ is given in Figure~\ref{fig:Machine}.
% operational semantics, consisting in the usual weak call-by-value reduction extended with the two non-deterministic rules \choiceL, \choiceR.
We denote by $\msto$ the transitive and reflexive closure of $\to$.
\end{definition}

The side condition $(\ast)$ on the context rules for the application avoids critical pairs with the $\parallel$-rules: this is not actually needed but it simplifies some proofs. 
A term $M$ is called a \emph{normal form} if there is no $N\in\Terms$ such that $M\to N$.
In particular, all (parallel compositions of) values are normal forms.
Note that when $M$ is closed then either it is a parallel composition of values or it reduces.
\begin{definition}\label{def:converging} A closed term $M\in\Terms$ \emph{converges} if and only if there exists a reduction $M\msto V_1\parallel \cdots\parallel V_n$ for some $V_i\in\Values$.
\end{definition}
The intuitive idea underlying the above notion of convergence is the following:
\begin{itemize}
\item The non-deterministic choice $M + N$ is treated as \emph{may}-convergent, either of the alternatives may be chosen during the reduction and the sum converges
	if either $M$ or $N$ does.
\item The parallel composition $M\parallel N$ is modelled as \emph{must}-convergent, the reduction forks and the parallel composition converges if both $M$ and $N$ do.
\end{itemize}

Let us provide some examples. %\todog{don't write $\to$-normal forms. Normal forms have been defined before wr.t. $\to$}
We set $\mathbf{I} = \lam x.x$, $\mathbf\Delta=\lambda x.xx$ and we denote by $\mathbf\Omega$ the paradigmatic non-converging term $\mathbf\Delta\mathbf\Delta$, 
which reduces to itself as $\mathbf\Delta$ is a value. The reduction is \emph{lazy}, i.e.\ it does not reduce under abstractions, 
so for example $\lambda y.\mathbf\Omega$ is a normal form. 
In fact, when considering closed terms, the parallel compositions of values are exactly the normal forms, thus justifying Definition~\ref{def:converging}. 
We would like to stress that our system is designed in such a way that a parallel composition of values is not a value. 
As a consequence, the term $P=\lambda k.\mathbf\Delta\parallel\mathbf\Delta$ is not a value, so the term $(\lambda x.x\mathbf Ix)P$ is converging. 
Indeed, it reduces to $(\lambda x.x\mathbf{I}x)(\lam k.\mathbf{\Delta})\parallel(\lambda x.x\mathbf{I}x)\mathbf{\Delta}\to^\ast\mathbf\Delta\parallel\mathbf\Delta$. 
Notice that, if we consider $P$ as a value, then $(\lambda x.x\mathbf{I}x)P$ would diverge since it would reduce to $P\mathbf IP\to^\ast(\mathbf\Delta \parallel\mathbf I)P\to^\ast\mathbf\Delta P\parallel P$ and one can check easily that $\mathbf\Delta P$ diverges.
%\todog{the subtle point is that a parallel composition of VALUES is not a value} $(\lambda y.\mathbf{I})(\mathbf\Omega\parallel \mathbf{I})$, where $\mathbf{I} = \lam x.x$, does not reduce to $\mathbf{I}$ but to $(\lambda y.\mathbf{I})\mathbf\Omega\parallel(\lambda y.\mathbf{I})\mathbf{I}\to(\lambda y.\mathbf{I})\mathbf\Omega\parallel \mathbf{I}\to(\lambda y.\mathbf{I})\mathbf\Omega\parallel \mathbf{I}\to\dots$ so looping.

The presence of the non-deterministic choice $+$ enlightens a typical feature of the call-by-value $\lambda$-calculus: 
application is bilinear (i.e.\ it commutes with $+$) while abstraction is not linear. 
Indeed, one can prove that $(M+M')(N+N')$ and $MN+MN'+M'N+M'N'$ are operationally indistinguishable, %\footnote{This can be formally proved as a corollary of Theorem~\ref{thm:Convergence}, using the remark that the two terms have the same types.}, 
while $\lambda x.(M+N)$ and $\lambda x.M+\lambda x. N$, in general, are not. 
For example, take $S=\lambda x.(x+\mathbf I)$, $S'=\lambda x.x+\lambda x.\mathbf I$, $E_{\mathbf I}=\lambda x.\mathbf I$, $E_{\mathbf \Omega}=\lambda x.\mathbf \Omega$, and $F=\lambda b.bE_{\mathbf\Omega}(bE_{\mathbf I}E_{\mathbf \Omega})\mathbf I$.
Now observe that $FS$ is converging to the value $\mathbf I$, while $FS'$ diverges. Indeed, remarking that $SE_{\mathbf I}E_{\mathbf\Omega}$ reduces non-deterministically to $\mathbf I$ and to $E_{\mathbf\Omega}$, we have:
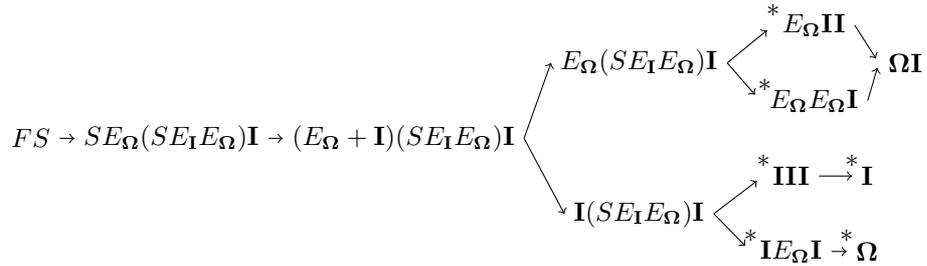
\begin{center}
\begin{tikzpicture}
\node (one) at (0,0) {$FS$};
\node (two) at (1.85,0) {$SE_{\mathbf\Omega}(SE_{\mathbf I}E_{\mathbf \Omega})\mathbf I$};
\node (twotwo) at (4.9,0) {$(E_{\mathbf\Omega}+\mathbf I)(SE_{\mathbf I}E_{\mathbf \Omega})\mathbf I$};
%	\node (twoup) at (4.8,1) {$\mathbf IE_{\mathbf\Omega}(SE_{\mathbf I}E_{\mathbf \Omega})\mathbf I$};
	\node (threeup) at (8,1) {$E_{\mathbf\Omega}(SE_{\mathbf I}E_{\mathbf \Omega})\mathbf I$};
		\node (threeupup) at (10.3,1.5) {$E_{\mathbf\Omega}\mathbf I\mathbf I$};
		\node (threeupdown) at (10.3,.5) {$E_{\mathbf\Omega}E_{\mathbf\Omega}\mathbf I$};
	\node (fourup) at (11.5,1) {$\mathbf\Omega\mathbf I$};	
%	\node (twodown) at (4.8,-1) {$E_{\mathbf I}E_{\mathbf\Omega}(SE_{\mathbf I}E_{\mathbf \Omega})\mathbf I$};
	\node (threedown) at (8,-1) {$\mathbf I(SE_{\mathbf I}E_{\mathbf \Omega})\mathbf I$};
		\node (threedownup) at (10,-.5) {$\mathbf I\mathbf I\mathbf I$};
		\node (fourdownup) at (11,-.5) {$\mathbf I$};
		\node (threedowndown) at (10,-1.5) {$\mathbf IE_{\mathbf \Omega}\mathbf I$};
		\node (fourdowndown) at (11,-1.5) {$\mathbf \Omega$};	
\draw [->] (one.east) -- (two.west);
\draw [->] (two.east) -- (twotwo.west);
\draw [->] (twotwo.east) -- (threeup.west);
	\draw [->, shorten >=3pt] (threeup.east) -- (threeupup.west) node [above] {$\ast$};
		\draw [->, shorten >=2pt] (threeupup.east) -- (fourup.west);
	\draw [->, shorten >=3pt] (threeup.east) -- (threeupdown.west) node [above] {~$\ast$};
		\draw [->, shorten >=2pt] (threeupdown.east) -- (fourup.west);
\draw [->] (twotwo.east) -- (threedown.west);
	\draw [->, shorten >=3pt] (threedown.east) -- (threedownup.west) node [above] {$\ast$};
		\draw [->] (threedownup.east) -- (fourdownup.west) node [above] {$\ast$};
	\draw [->, shorten >=3pt] (threedown.east) -- (threedowndown.west) node [above] {$\ast$};
		\draw [->] (threedowndown.east) -- (fourdowndown.west) node [above] {$\ast$};
\end{tikzpicture}
\end{center}
while $FS'$ has two reducts, either $F\mathbf I$ reducing to $\mathbf\Omega\mathbf I$, or $FE_{\mathbf I}$
reducing to $\mathbf\Omega$.

Finally, we give two examples mixing $+$ and $\parallel$~. The term $(\lambda x.(x\parallel x))(V+V')$ converges either to $V\parallel V$ or to $V'\parallel V'$, while the term $(\lambda x.(x+x))(V\parallel V')$ converges to $V\parallel V'$, only. %Finally, the term $(\lambda x.(x\parallel x))(\lambda k.(V+V'))$

% !TEX root = ../main.tex
%
\section{Linear Logic Based Type System}\label{sect:types}

In this section we introduce our type system based on linear logic.
The set $\Types$ of \emph{(parallel) types} and the set $\LTypes$ of \emph{computational types} are generated by the following grammar:
$$
\begin{array}[t]{l@{\hspace{1.5cm}}r@{\ ::=\quad}l@{\hspace{1.5cm}}l}
      \textrm{\itshape\firstsort{s}:} & \alpha,\beta & \alpha\parr\beta~|~\tau&\Types\\
      \textrm{\itshape\secondsort{s}:} & \tau,\rho & \one~|~\tau\otimes\rho~|~\tau\multimap\alpha&\LTypes\\
\end{array}
$$
For the sake of simplicity, types are considered up to associativity and commutativity of the tensor $\otimes$ and the par $\parr$. %\todom{rev4 asks why we need this assumption, especially because we do not need the same for $\parallel$ and $+$...}. %\todoa{If $+$ and $\parallel$ are not AC, probably $\otimes$ and $\parr$ does not need it neither. We must verify all the proofs} 
%(This is not strictly needed, but allow cleaner notations, namely the definition of $n$-ary $\otimes$ and $\parr$. In fact, the $\multimap_E$-rule of Figure~\ref{fig:Types} would definitely have a clumsy shape without such $n$-ary connectives.)
The type $\one$, which is the only atomic type, represents the empty tensor and
is therefore its neutral element (i.e.~$\tau\otimes\one = \tau$).
Accordingly, we write $\otimes_{i=1}^n\tau_i$ for $\tau_1\otimes\cdots\otimes\tau_n$ when $n\geq 1$, and for $\one$ when $n=0$.
Similarly,  when $n\geq 1$, $\parr_{i=1}^n\alpha_i$ stands for $\alpha_1\parr\cdots\parr\alpha_n$. 

As mentioned in the introduction, $\tau_1\otimes\cdots\otimes\tau_n$ and $\alpha_1\parr\cdots\parr\alpha_k$ are actually notations representing two different kinds of multisets, namely the $\oc$- and $\wn$-multisets (respectively).
Under this correspondence, $\one$ represent the empty $\oc$-multiset.
We do not allow the empty par as it would correspond to an empty sum of terms, that would be delicate to treat operationally (cf.\ \cite{ArrighiDowekRTA08}).
%\todom{rev4 asks why we do not have the empty par. I would answer "because it would correspond to the empty sum at the term level. we do not have the empty sum, because that gives more delicate issue in the operational semantics.}\todoa{Agreed}
%We will also shorten $\one\otimes\cdots\otimes\one$ ($n$ times) into $\otimes^n \one$ and $\one\parr\cdots\parr\one$ ($k> 0$ times) writing $\parr^k \one$.\todog{added}~

Note that neither $\otimes$ nor $\parr$ are supposed idempotent.

\begin{figure}[!t]
  \begin{center}
   $\begin{array}{c}
   ~\\
    {
    \prooftree
    \justifies x:\tau\vdash x:\tau
    \using ax
    \endprooftree
        \hspace{1cm}
    \prooftree\Delta_i,x:\tau_i\vdash M:\alpha_i\qquad 1\leq i\leq n
    \justifies\bigotimes\limits_{i=1}^n\Delta_i\vdash\lambda x.M:\bigotimes\limits_{i=1}^n(\tau_i\multimap\alpha_i)
    \using\multimap_I\quad n\geq 0
    \endprooftree
    }\\
    \\
    {
    \prooftree\Delta\vdash M:\bigparr\limits_{i=1}^k\bigotimes\limits_{j=1}^{n_i}(\tau_{ij}\multimap\alpha_{ij})\qquad\Gamma_i\vdash N:\bigparr\limits_{j=1}^{n_i}\tau_{ij}\quad 1\leq i\leq k
    \justifies\Delta\otimes\bigotimes\limits_{i=1}^{k}\Gamma_i\vdash MN:\bigparr\limits_{i=1}^k\bigparr\limits_{j=1}^{n_i}\alpha_{ij}
    \using\multimap_E\quad \begin{array}{c}k\geq 1\\ n_i\geq 1\end{array}
    \endprooftree
    }\\
    \\
    \prooftree\Delta\vdash M:\alpha
    \justifies\Delta\vdash M+N:\alpha
    \using \plusL
    \endprooftree
    \hspace{0.5cm}
    \prooftree\Delta\vdash N:\alpha
    \justifies\Delta\vdash M+N:\alpha
    \using \plusR
    \endprooftree
	  \hspace{0.5cm}
    \prooftree\Delta\vdash M:\alpha_1\qquad\Gamma\vdash N:\alpha_2
    \justifies\Delta\otimes\Gamma\vdash M\parallel N:\alpha_1\parr\alpha_2
    \using \parallel_I
    \endprooftree\\
    \end{array}$
  \end{center}
  \caption{Type system: the inference rules.}
  \label{fig:Types}
\end{figure}

\begin{definition}
A \emph{context} $\Gamma$ is a total map from $\Var$ to $\LTypes$, such that $\dom(\Gamma) = \{ x \mid \Gamma(x)\neq\one\}$ is finite. 
The tensor of two contexts $\Gamma$ and $\Delta$, written $\Gamma\otimes\Delta$, is defined pointwise. %by $(\Gamma\otimes\Delta)(x) = \Gamma(x)\otimes\Delta(x)$, for all $x\in\Var$.
\end{definition}

As a matter of notation, we write $x_1:\tau_1,\dots, x_n:\tau_n$ for the context $\Gamma$ such that $\Gamma(x_i)=\tau_i$ and $\Gamma(y) = \one$ for all $y\notin\vec x$.
The context mapping all variables to $\one$ is denoted by $\emptyset$; note that $\Gamma\otimes\emptyset = \Gamma$.

\begin{definition}
\begin{itemize}
\item
	The \emph{type system} for $\Terms$ is defined in Figure~\ref{fig:Types}.
	\emph{Typing judgements} are of the form $\Gamma\vdash M:\alpha$;
	when $\Gamma = \emptyset$ we simply write $\ \vdash M : \alpha$.
	\emph{Derivation trees} will be denoted by $\pi$.
\item
	A term $M\in\Terms$ is \emph{typable} if there exist $\alpha\in\Types$ and a context $\Gamma$ such that $\Gamma\vdash M : \alpha$.
\end{itemize}
\end{definition}

The rules for typing non-deterministic choice and parallel composition reflect their operational behaviour.
Non-deterministic choice is may-convergent, thus it is enough to ask that one of the terms in a sum is typable;
on the other hand parallel composition is must-convergent, we therefore require that all its components are typable. 
Intuitively, when dealing with closed terms, the $\parr$ operator can be only introduced to type a parallel composition, and gives an account of the number of its components.
In fact, for closed regular \lam-terms, the type system looses the $\parr$-level and %morally 
 collapses to the one presented in~\cite{Ehrhard12}. %Namely, the $\multimap_E$ rule has the parameters $k, n_i=1$ whenever the $\parallel$ operator does not occur in the terms $M$ and $N$.

The $\multimap_E$ rule reflects the distribution of the parallel operator over the application. For example, take $M=x\parallel x'$ and $N=y\parallel y'$ in the premises of $\multimap_E$, then we have $k=2$ and $n_1=n_2=2$ so that the type of the term $MN$ is a $\parr$ of four types, which is in accordance with $(x\parallel x')(y\parallel y')\to^\ast(xy\parallel xy')\parallel(x'y\parallel x'y')$.

\begin{myremark}\label{rmk:valuesTop}
 For every $V\in\Values$ we can derive $\vdash V:\one$. Indeed, if $V$ is a variable, then the derivation follows by $ax$; if $V$ is an abstraction, then it follows by $\multimap_I$ using $n=0$.
 As a simple consequence we get $\ \vdash V_1\parallel \cdots\parallel V_k :  \one\parr\cdots\parr\one$ ($k$~times) for all $V_1,
 \dots, V_k\in\Values$.
\end{myremark}

Concerning the possible types of values, the next more general lemma holds.
\begin{lemma}\label{lem:values}
Let $V\in\Values$. If $\Delta\vdash V:\alpha$ then $\alpha\in\LTypes$.
\end{lemma}
\begin{proof} A proof of $\Delta\vdash V:\alpha$ ends in either a $ax$ or a $\multimap_I$ rule. In both cases $\alpha$ is a \secondsort. \qed
%By cases on the shape of $V$. If $V=x$, then $\Delta=x:\alpha$ and $\alpha$ must be a \secondsort.
%Otherwise $V=\lambda x.M$ and $\alpha=\bigotimes_{i=1}^n(\rho_i\multimap\alpha_i)$, which is a \secondsort{} too.
\end{proof}

To help the reader to get familiar with the type system, we provide some examples of typable and untypable terms.

\begin{example} Recall that $\mathbf{I} = \lam x.x$, $\mathbf{\Delta} = \lam x.xx$ and $\mathbf{\Omega} = \mathbf{\Delta}\mathbf{\Delta}$.
\begin{enumerate}
\item $\vdash \mathbf{I} : \bigotimes_{i=1}^n (\tau_i\multimap \tau_i)$ and $\ \vdash \lam x.\mathbf{I} : \bigotimes_{i=1}^n(\one\multimap\bigotimes_{j=1}^{k_i} (\tau_{ij}\multimap \tau_{ij}))$.
\item $\vdash \mathbf{\Delta} : \bigotimes_{i=1}^n ((\tau_i\multimap \alpha_i)\otimes \tau_i)\multimap\alpha_i$.
\item %\todom{extended to satisfy rev2} 
	$\mathbf\Omega$ is not typable. By contradiction, suppose $\vdash\mathbf\Omega:\alpha$. By $(\multimap_E)$ and (2) there is a type $\tau$ such that $\vdash\mathbf{\Delta} : \tau\multimap\alpha$ and $\vdash\mathbf\Delta : \tau$. Let us choose such a $\tau$ with minimal size. Applying (2) to $\vdash\mathbf\Delta : \tau\multimap\alpha$, we get $\tau=(\tau'\multimap\alpha)\otimes\tau'$, from which one can deduce (see Lemma~\ref{lem:intersec}, below) that $\vdash\mathbf\Delta : \tau'\multimap\alpha$ and $\vdash\mathbf\Delta : \tau'$, thus contradicting the minimality of $\tau$. 
\item However, $\vdash \lam x.\mathbf\Omega : \one$, so $\vdash \lam x.\mathbf\Omega + \mathbf\Omega : \one$, but $\lam x.\mathbf\Omega \parallel \mathbf\Omega$ is not typable. 
\item From (1) and (4) we get: $\vdash \mathbf{I} \parallel \lam x.\mathbf\Omega :  (\bigotimes_{i=1}^n (\tau_i\multimap \tau_i))\parr \one$.
\end{enumerate}
\end{example}

We now define a measure associating a natural number with every derivation tree.
In Section~\ref{subsec:SR} we prove that such a measure decreases along the reduction.
In the next definition we follow the notation of Figure~\ref{fig:Types}, 
in particular in the $\multimap_E$-case the parameter $n_i$ refers to the arity of the $\bigparr$ in the conclusion of $\pi_i$.

\begin{definition}%\todom{corrected third case definition as required by rev2}
% Let $\pi$ be a derivation tree. We define a measure $|\pi|$ inductively as follows:
The \emph{measure} $\mes{\pi}$ of a derivation tree $\pi$ is defined inductively as: %(\todom{added for rev 2 and for rev4}recall the notation of~Figure~\ref{fig:Types}, in particular in the $\multimap_E$-case the parameter $n_i$ refers to $\bigparr$'s arity in the conclusion of $\pi_i$):
\smallskip 

 \begin{tabular}{p{5cm}p{6.5cm}}
  $\pi=\prooftree\hspace{1cm}
	   \justifies S
	   \using ax
	   \endprooftree$
  &
  $|\pi|=0$
  \\[1.1em]
  $\pi=\prooftree\pi_1\ \cdots\ \pi_n
	   \justifies S
	   \using\multimap_I
	   \endprooftree$
  &
  $|\pi|= \sum_{i=1}^n|\pi_i|$
  \\[1.1em]
  $\pi=\prooftree\pi_0\quad\pi_1\dots\pi_k
	   \justifies S
	   \using\multimap_E\ \begin{array}{c}k\geq 1\\ n_i\geq 1\end{array}
	   \endprooftree$
  &
  $|\pi|=\sum_{i=0}^k|\pi_i|+(\sum_{i=1}^k2n_i)-1$%\textrm{ (cf.~Fig.~\ref{fig:Types})}
  \\[1.1em]
%  \multicolumn{2}{l}{(cf.~the notation of~Fig.~\ref{fig:Types}, in particular $n_i$ refers to $\bigparr$'s arity in the}
%  \\
%  \multicolumn{2}{l}{}
%  \\[1.1em]
  %%
  $\pi=\prooftree\pi'
	   \justifies S
	   \using\plusL
	   \endprooftree$
	   \quad or\quad
  $\pi=\prooftree\pi'
	   \justifies S
	   \using\plusR
	   \endprooftree$
  &
  $|\pi|=|\pi'|+1$
  \\[1.1em]
  $\pi=\prooftree\pi_1\qquad\pi_2
	   \justifies S
	   \using \parallel_I
	   \endprooftree$
  &
  $|\pi|=|\pi_1|+|\pi_2|$
 \end{tabular}
 \medskip

Hereafter, we may slightly abuse the notation and write $\pi=\Gamma\vdash M:\alpha$ to refer to a derivation tree $\pi$ ending by the sequent $\Gamma\vdash M:\alpha$.
 \end{definition}

The measure of a derivation only depends on its rules of type $\multimap_E$, $\plusL$ and $\plusR$. 
These are in fact the kinds of rules that can type a redex ($\beta_v$ and $\parallel$ redexes are typed by $\multimap_E$ rules, 
$+$ redexes by $\plusL$, $\plusR$ rules). 
Each occurrence of a $\plusL$ or $\plusR$ rule counts for one, because a $+$-reduction does not create new rules in the derivation typing 
the contractum (see the proof of Theorem~\ref{thm:SR} for more details). 
An occurrence of a $\multimap_E$ counts for the number of ``active'' connectives appearing in the principal premise, i.e.\ the number of the connectives that are 
underlined in the left-most premise of the $\multimap_E$ rule in Figure~\ref{fig:Types}, indeed
$$
 	\underbrace{\sum_{i=1}^kn_i}_{\textrm{$\multimap$'s}}+\underbrace{\sum_{i=1}^k(n_i-1)}_{\textrm{$\otimes$'s}}+\underbrace{(k-1)}_{\textrm{$\parr$'s}}=(\sum_{i=1}^k2n_i)-1.
$$
Such a weight is needed since the $\parallel$-reduction creates two new $\multimap_E$ rules in the derivation typing the contractum. The measure decreases however, since the sum of the weight of the two new rules is less than the weight of the eliminated rule. 
\begin{figure}
\begin{align*}
\pi&=\quad
\prooftree
	\prooftree
		\pi_1=x:\tau\vdash xx:\one
	\quad
		\pi_2=x:\tau\vdash xx: \one
	\justifies \vdash\mathbf{\Delta}:(\tau\multimap\one)\otimes(\tau\multimap\one)
	\using \multimap_I
	\endprooftree
\quad
	\prooftree
	\vdash\mathbf I:\tau\quad\vdash\lambda xy.\mathbf\Omega:\tau
	\justifies
	\vdash\mathbf I\parallel\lambda xy.\mathbf\Omega:\tau\parr\tau
	\using \parallel_I
	\endprooftree
\justifies
\vdash\mathbf\Delta(\mathbf I\parallel\lambda xy.\mathbf\Omega):\one\parr\one
\using \multimap_E
\endprooftree\\[20pt]
\pi'&=\
\prooftree
	\prooftree
		\prooftree
			\pi_1=x:\tau\vdash xx:\one
			\justifies \vdash\mathbf{\Delta}:\tau\multimap\one
			\using \multimap_I
		\endprooftree
		\;
		\vdash\mathbf I:\tau
	\justifies \vdash\mathbf{\Delta}\mathbf I:\one
	\using\multimap_E	
	\endprooftree
\;
	\prooftree
		\prooftree
			\pi_2=x:\tau\vdash xx:\one
			\justifies \vdash\mathbf{\Delta}:\tau\multimap\one
			\using \multimap_I
		\endprooftree
		\;
		\vdash\lambda xy.\mathbf\Omega:\tau
	\justifies \vdash\mathbf{\Delta}(\lambda xy.\mathbf\Omega):\one
	\using\multimap_E	
	\endprooftree
\justifies
\mathbf\Delta\mathbf I\parallel\mathbf\Delta(\lambda xy.\mathbf\Omega):\one\parr\one
\using\parallel_I
\endprooftree
\end{align*}
\caption{Derivation trees typing, respectively, the $\parallel$-redex $\mathbf\Delta(\mathbf I\parallel\lambda xy.\mathbf\Omega)$ and its contractum $\mathbf\Delta\mathbf{I}\parallel \mathbf\Delta(\lambda xy.\mathbf\Omega)$, taking $\tau=(\one\multimap\one)=(\one\multimap\one)\otimes\one$.}\label{fig:example_typing}
\end{figure}

 For example, let us consider the derivation tree $\pi$ in Figure~\ref{fig:example_typing}, which types the $\parallel$-redex $\mathbf\Delta(\mathbf I\parallel\lambda xy.\mathbf\Omega)$ with $\one\parr\one$, and has three $\multimap_E$ rules --- 
 one of weight 1 in each subtree $\pi_1$, $\pi_2,$ and one of weight 3 giving the conclusion, so that $\mes{\pi} = 5$. 
 %The weights of the formers are $1$ and the one of the latter is 3, so that $|\pi|=5$. 
% We have that $\mathbf\Delta(\mathbf I\parallel\lambda xy.\mathbf\Omega)\to \mathbf\Delta\mathbf I\parallel\mathbf\Delta(\lambda xy.\mathbf\Omega)$ and the $\multimap_E$-rule ending $\pi$ splits in two ``residues'' in the derivation tree $\pi'$ giving $\mathbf\Delta\mathbf I\parallel\mathbf\Delta(\lambda xy.\mathbf\Omega):\one\parr\one$. 
Now, the $\multimap_E$-rule ending $\pi$ splits into two $\multimap_E$-rules in the derivation tree $\pi'$ typing the contractum of $\mathbf\Delta(\mathbf I\parallel\lambda xy.\mathbf\Omega)$, namely $\pi'=~\vdash\mathbf\Delta\mathbf I\parallel\mathbf\Delta(\lambda xy.\mathbf\Omega):\one\parr\one$. 
However, $|\pi'|=|\pi|-1$ since %the sum of 
the number of the active connectives of the $\multimap_E$-rule concluding $\pi$ is greater than the sum of the number of the active connectives of its ``residuals'' in $\pi'$. 

Finally, note that the term $\mathbf\Delta(\mathbf I\parallel\lambda xy.\mathbf\Omega)$ reduces to the value $\mathbf I\parallel\lambda y.\mathbf\Omega$ in $5=|\pi|$ steps. As we will show in Theorem~\ref{thm:Convergence} this does not happen by chance.
\section{Properties of the Type System}\label{sec:properties}

%We prove that the set of types assigned to a term is invariant under $\to$. Let $M$ a term and $\{N_1, N_2\}$, Theorem~\ref{thm:SR} states that any type of $M$ is a type of $N$ (\emph{subject reduction}), while Theorem~\ref{thm:SE} achieves the converse, any type of $N$ is a type of $M$ (\emph{subject expansion}). Moreover, Theorem~\ref{thm:SR} proves that the measure associated with the typing tree of a term strictly decreases at each step of the reduction. This is typical of the non-idempotent intersection type systems, as discussed in the Introduction, and has as a trivial consequence that any typable term $M$ is $\to$-normalising and the measure associated with any type derivation of $M$ bounds the length of any $\to$-reduction sequence starting from $M$.\todom{not precise, problem of $+$ reduction}

%We prove that the set of types assigned to a term has an invariant reduction path under~$\to$. 
%Given $M\to N$, Theorem~\ref{thm:SR} states that either any type of $M$ is a type of $N$ or there exists an alternative reduction $M\to N'$, using $+$-rules, 
%where any type of $M$ is a type of $N'$ (\emph{subject reduction}). 
%On the other hand Theorem~\ref{thm:SE} shows the converse, namely that any type of $N$ is a type of $M$ (\emph{subject expansion}). 
%Moreover, the two theorems combined prove that the measure associated with the typing tree of a term decreases (resp.\ increases) of exactly one unity 
%at each typed step of reduction (resp.\ expansion). 

We prove that the set of types assigned to a term is invariant under~$\to$, in a non-deterministic setting. 
More precisely, Theorem~\ref{thm:SR} states that if $N$ is the contractum of a $\{\beta_v,\parallel\}$-redex in $M$, then any type of $M$ is a type of $N$, and if $N$ and $N'$ are the two possible contracta  of a $+$-redex in $M$, then any type of $M$ is either a type of $N$ or of $N'$ (\emph{subject reduction}). 
On the other hand Theorem~\ref{thm:SE} shows the converse, namely that whenever $M\to N$, any type of $N$ is a type of $M$ (\emph{subject expansion}). 

Moreover, the two theorems combined prove that the measure associated with the typing tree of a term decreases (resp.\ increases) of exactly one unit 
at each typed step of reduction (resp.\ expansion). This is typical of non-idempotent intersection type systems, as discussed in the introduction. 
As a consequence, any typable term $M$ is normalising and the measure of 
specific derivation trees of $M$ gives the length of a converging reduction sequence.

\subsection{Subject reduction}\label{subsec:SR}

In order to prove subject reduction we first need some preliminary lemmas.
Their proofs are lengthy but not difficult, therefore we write explicitly only the most interesting cases.
% --- the remaining cases can be found in Appendix~\ref{app:SR}.

\begin{lemma}\label{lem:intersec}%\todom{how to answer to rev4 ``the notation $\pi=...$ is misleading, since sometimes you use it as a universal quantification, sometimes as existential quantification'' ? (modified some statements when easy)}
We have that $\pi=\Delta\vdash V:\bigotimes_{i=1}^n\tau_i$ if and only if $\Delta=\bigotimes_{i=1}^n\Delta_i$ and
$\pi_i = \Delta_i\vdash V:\tau_i$ for all $i=1,\dots,n$.
Moreover, $\mes{\pi} = \sum_{i=1}^n \mes{\pi_i}$.
\end{lemma}

\begin{proof} We only prove $(\Rightarrow)$, %the left-to-right implication, 
the other direction being similar.
Since $V$ is a value, the last rule of $\pi$ is either $ax$ or $\multimap_I$. The first case is trivial. In the second case, $V=\lambda x. M$ and the premises of the $\multimap_I$-rule are $m\geq n$, say $\pi_j'=\Delta_j,x:\rho_j\vdash M:\alpha_j$ for $j\leq m$, and $\tau_1=\bigotimes_{j=1}^{m_1}\rho_j\multimap\alpha_j$ and $\Delta_1=\bigotimes_{j=1}^{m_1}\Delta_j$, \dots, $\tau_n=\bigotimes_{j=m_{n-1}+1}^{m_n}\rho_j\multimap\alpha_j$  and $\Delta_n=\bigotimes_{j=m_{n-1}+1}^{m_n}\Delta_j$, with $m_1+\dots+m_n=m$.

Notice $|\pi|=\sum_{j=1}^m|\pi'_j|$. Then, for every $i\leq n$, a $\multimap_I$-rule with premises $\pi'_{m_{i-1}+1},\dots,\pi'_{m_i}$ yields $\pi_i=\Delta_i\vdash\lambda x.M:\tau_i$, with $|\pi_i|=\sum_{j=m_{i-1}+1}^{m_i}|\pi_i'|$, therefore $|\pi|=\sum_{i=1}^n|\pi_i|$.\qed
\end{proof}

%The substitution lemma below will be the key ingredient in the proof of subject reduction.
%It ensures that when substituting a value for a free variable (of the same type) in a well-typed term its type judgment remains valid, provided that the judgment contexts of the term and of the value are combined together.
%
\begin{lemma}[Substitution lemma]\label{lem:substitution}
 If $\pi_1=\Delta,x:\tau\vdash M:\alpha$ and $\pi_2=\Gamma\vdash V:\tau$, then there is $\pi_3=\Delta\otimes\Gamma\vdash M[V/x]:\alpha$.
 Moreover $\mes{\pi_3}=\mes{\pi_1} + \mes{\pi_2}$.
\end{lemma}
\begin{proof}
  By structural induction on $M$. We only treat the most interesting case, namely $M=NP$. 
  In this case, the last rule of $\pi_1$ is a $\multimap_E$-rule with $k+1$ premises, say
   $\pi_1^0=\Delta_0,x:\tau_0\vdash N:\bigparr_{i=1}^k\bigotimes_{j=1}^{n_i}(\rho_{ij}\multimap\alpha_{ij})$,
   and for $i=1,\dots,k$,
   $\pi_{1}^i=\Delta_i,x:\tau_i\vdash P:\bigparr_{j=1}^{n_i}\rho_{ij}$,
   where
   $\Delta=\bigotimes_{i=0}^k\Delta_i$,
   $\tau=\bigotimes_{i=0}^k\tau_i$,
   $\alpha=\bigparr_{i=1}^k\bigparr_{j=1}^{n_i}\alpha_{ij}$ and
   $|\pi_1|=\sum_{i=0}^k|\pi_{1}^i|+(\sum_{i=1}^k2n_i)-1$.
   By Lemma~\ref{lem:intersec}, we can split $\pi_2$ into $k+1$ derivations $\pi_{2}^i=\Gamma_i\vdash V:\tau_i$, for $i= 0,\dots, k$, such that $\Gamma=\bigotimes_{i=0}^k\Gamma_i$ and $|\pi_2|=\sum_{i=0}^k|\pi_{2}^i|$.
   By the induction hypothesis, there are
   $\pi_{3}^0=\Delta_0\otimes\Gamma_0\vdash N[V/x]:\bigparr_{i=1}^k\bigotimes_{j=1}^{n_i}(\rho_{ij}\multimap\alpha_{ij})$,
   with
   $|\pi_{3}^0|=|\pi_{1}^0|+|\pi_{2}^0|$,
   and for $i=1,\dots,k$,
   $\pi_{3}^i=\Delta_i\otimes\Gamma_i\vdash P[V/x]:\bigparr_{j=1}^{n_i}\rho_{ij}$,
   with
   $|\pi_{3}^i|=|\pi_{1}^i|+|\pi_{2}^i|$.
   Hence, by rule $\multimap_E$, we have
   \begin{equation*}
   \pi_3=(\Delta_0\otimes\Gamma_0)\otimes\bigotimes_{i=1}^k(\Delta_i\otimes\Gamma_i) \vdash N[V/x]P[V/x]:\bigparr_{i=1}^k\bigparr_{j=1}^{n_i}\alpha_{ij}
   \end{equation*}
   Notice that
   $(\Delta_0\otimes\Gamma_0)\otimes\bigotimes_{i=1}^k(\Delta_i\otimes\Gamma_i)=\Delta\otimes\Gamma$ and $N[V/x]P[V/x]=(NP)[V/x]$.
   Moreover,
   $|\pi_3|
   =\sum_{i=0}^k|\pi_{3}^i|+(\sum_{i=1}^k2n_i)-1
   =\sum_{i=0}^k(|\pi_{1}^i|+|\pi_{2}^i|)$ $+(\sum_{i=1}^k2n_i)-1
   =(\sum_{i=0}^k|\pi_{1}^i|+(\sum_{i=1}^k2n_i)-1)+\sum_{i=0}^k|\pi_{2}^i|
   =|\pi_1|+|\pi_2|$.\qed
\end{proof}

We now prove the subject reduction property, which ensures that the type is preserved during reduction, while the measure of the typing is strictly decreasing.

As a matter of terminology, we say that a term $M$ \emph{reduces to a term $N$ using $+$-reductions}, if $M\to N$ is derivable as a direct consequence of a $+$-reduction and
(possibly) some contextual rules. 
In the following proof, given a set $S$, we denote by $\sharp S$ its cardinality.

%We say that a term $M$ reduces to a term $N$ using $+$-reductions, if either its reduction is a direct use of a $+$-reduction or a consequence of a $+$-reduction and a contextual rule.\todom{let us remove this sentence?}\todoa{If in ``reduces using $+$-reductions'' it is clear enough that it considers also contexts, then yes, we can}

\begin{theorem}[Subject reduction]\label{thm:SR} Let $\pi=\Delta\vdash M:\alpha$.
\begin{itemize}
\item 
	If $M\to N$ without using $+$-reductions, then there is $\pi'=\Delta\vdash N:\alpha$.
\item 
	If $M\to N_1$ and $M\to N_2$ using $+$-reductions, then there is $\pi'$ such as either $\pi'=\Delta\vdash N_1:\alpha$ or $\pi'=\Delta\vdash N_2:\alpha$.
\end{itemize}
Moreover, in both cases we have $|\pi'|=|\pi|-1$.
\end{theorem}
\begin{proof}
 We proceed by induction on the length of the derivation of $M\to N$. 
 We only treat the most interesting cases.
 %The contextual cases are omitted.
 \begin{itemize}
  \item $(\lambda x.M')V\to M'[V/x]$.
  	  Then, the last rule of $\pi$ is a $\multimap_E$-rule with $k+1$ premises, say $\pi_0=\Delta'\vdash\lambda x.M':\bigparr_{i=1}^k\bigotimes_{j=1}^{n_i}(\rho_{ij}\multimap\alpha_{ij})$ and for every $i=1,\dots,k$, $\pi_{i}=\Gamma_i\vdash V:\bigparr_{j=1}^{n_i}\rho_{ij}$, with moreover $\Delta=\Delta'\otimes\bigotimes_{i=1}^k\Gamma_i$, $\alpha=\bigparr_{i=1}^k\bigparr_{j=1}^{n_i}\alpha_{ij}$,
	  and $\mes{\pi} = \sum_{i=0}^k \mes{\pi_i} + (\sum_{i=1}^k 2n_i) -1$.
	  However, since Lemma~\ref{lem:values} entails that $k=n_1=1$ we get $\mes{\pi} = \mes{\pi_0} + \mes{\pi_1} + 1$.
   In addition, the only possibility for $\pi_0$ is to come from $\pi'_0=\Delta',x:\rho\vdash M':\alpha$, where $|\pi_0|=|\pi'_0|$. By Lemma~\ref{lem:substitution},
   $\pi'=\Delta'\otimes\Gamma\vdash M'[V/x]:\alpha$, where
   $|\pi'|
   =|\pi'_0|+|\pi_1|
   =|\pi_0|+|\pi_1|
   =|\pi|-1$.
   We conclude since $\Delta'\otimes\Gamma=\Delta$.

  \item Let $V(M\parallel N)\to VM\parallel VN$. Then
  $\pi=\Delta\otimes\bigotimes_{i=1}^k\Gamma_i\vdash V(M\parallel N):\bigparr_{i=1}^k\bigparr_{j=1}^{n_i}\alpha_{ij}$ ends in a $\multimap_E$ rule having as premises
  $\pi_0=\Delta\vdash V:\bigparr_{i=1}^k\bigotimes_{j=1}^{n_i}(\rho_{ij}\multimap\alpha_{ij})$
  and, for $i=1,\dots,k$,
  $\pi_i=\Gamma_j\vdash M\parallel N:\bigparr_{j=1}^{n_i}\rho_{ij}$. 
  Thus, we have $|\pi|=\sum_{j=0}^k|\pi_i|+(\sum_{i=1}^k2n_i)-1$. However, by Lemma~\ref{lem:values}, $k=1$, so we omit the index $i$ where it is not needed, and $|\pi|=|\pi_0|+|\pi_1|+2n-1$.
  Then
  $\pi_1^1=\Gamma_1\vdash M:\bigparr_{j\in S}\rho_j$
  and
  $\pi_1^2=\Gamma_2\vdash N:\bigparr_{j\in\bar{S}}\rho_j$,
  where $\Gamma=\Gamma_1\otimes\Gamma_2$, $\emptyset\neq S\subsetneq\{1,\dots,k\}$ and $\bar{S}=\{1,\dots,k\}\setminus S$
  with $|\pi_1|=|\pi_1^1|+|\pi_1^2|$. %\todom{next sentence, I have erased an useless reasoning}\todoa{Agreed}
    By Lemma~\ref{lem:intersec}, we can split $\pi_0$ into two derivations,
  $\pi_0^S=\bigotimes_{j\in S}\Delta_j\vdash V:\bigotimes_{j\in S}(\rho_j\multimap\alpha_j)$
  and
  $\pi_0^{\bar{S}}=\bigotimes_{j\in\bar{S}}\Delta_j\vdash V:\bigotimes_{j\in S}(\rho_j\multimap\alpha_j)$,
  with $|\pi_0^S|+|\pi_0^{\bar{S}}|=|\pi_0|$.
%  By Lemma~\ref{lem:intersec}, we can split $\pi_0$ in $n$ derivations,
%  $\pi_0^j=\Delta_j\vdash V:\rho_j\multimap\alpha_j$, for $j=1,\dots,n$, such that $|\pi_0|=\sum_{j=1}^n|\pi_0^j|$ and $\Delta=\bigotimes_{j=1}^n\Delta_j$. And by the same Lemma,
%  $\pi_0^S=\bigotimes_{j\in S}\Delta_j\vdash V:\bigotimes_{j\in S}(\rho_j\multimap\alpha_j)$
%  and
%  $\pi_0^{\bar{S}}=\bigotimes_{j\in\bar{S}}\Delta_j\vdash V:\bigotimes_{j\in S}(\rho_j\multimap\alpha_j)$,
%  with $|\pi_0^S|=\sum_{j\in S}|\pi_0^j|$ and
%  $|\pi_0^{\bar{S}}|=\sum_{j\in\bar{S}}|\pi_0^j|$.
  By rule $\multimap_E$, we have
  $\pi^1=\bigotimes_{j\in S}\Delta_j\otimes\Gamma_1\vdash VM:\bigparr_{j\in S}\alpha_j$
  and
  $\pi^2=\bigotimes_{j\in\bar{S}}\Delta_j\otimes\Gamma_2\vdash VN:\bigparr_{j\in\bar{S}}\alpha_j$,
  where $|\pi^1|=|\pi_0^S|+|\pi_1^1|+2\sharp S-1$, and
  $|\pi^2|=|\pi_0^{\bar{S}}|+|\pi_1^2|+2\sharp\bar{S}-1$.
  By rule $\parallel_I$,
  $\pi'=\bigotimes_{j=1}^n\Delta_i\otimes\Gamma_1\otimes\Gamma_2\vdash VM\parallel VN:\bigparr_{j=1}^n\alpha_j$, where
  $|\pi'|
  =|\pi^1|+|\pi^2|
  =(|\pi_0^S|+|\pi_1^1|+2\sharp S-1)+(|\pi_0^{\bar{S}}|+|\pi_1^2|+2\sharp\bar{S}-1)
  =|\pi_0|+|\pi_1|+2\sharp S+2\sharp\bar{S}-2
  =|\pi_0|+|\pi_1|+2n-2
  =|\pi|-1$.\qed
 \end{itemize}
\end{proof}

% !TEX root = ../main.tex
%
\subsection{Subject Expansion}\label{subsec:SE}

%The subject expansion is a standard property ensuring that any term reducing to a typable term is also typable.
The proof of the fact that our system enjoys subject expansion follows by straightforward induction, once one has proved 
the commutation of abstraction with abstraction, application, non-deterministic choice and parallel composition. 

%We refer to Appendix~\ref{app:SE} for a detailed proof.

%The subject expansion property, ensuring that any term reducing to a typable term is also typable, follows by straightforward induction from the above lemmas.
%Again, see Appendix~\ref{app:SE} for the detailed proof.

\begin{theorem}[Subject expansion]\label{thm:SE} If $M\to N$ and $\pi=\, \Delta\vdash N:\alpha$, then there is $\pi'=\, \Delta\vdash M:\alpha$, such that $\mes{\pi'}=\mes\pi+1$.
\end{theorem}

\begin{proof}
  By induction on the length of the derivation of $M\to N$, splitting into cases depending on its last rule. We only consider the most interesting case,
  i.e.\ $(\lambda x.M')V\to M'[V/x]$ where $M' = PQ$. One first needs to establish, by induction on $\pi$, a claim about the commutation of abstraction with application.
  
  \begin{claim} If $\pi=\,\Delta\vdash((\lambda x.P)V)((\lambda x.Q)V):\alpha$, where the last rule of $\pi$ is a $\multimap_E$ rule having $k+1$ premises, then there exists $\pi'=\,\Delta\vdash(\lambda x.PQ)V:\alpha$ such that $\mes{\pi'}=\mes{\pi}-k$.
  \end{claim}
  
  By definition we have $N=(PQ)[V/x]=P[V/x]Q[V/x]$. 
  So, $\pi=\,\Delta\vdash N:\alpha$ ends in a $\multimap_E$-rule with $k+1$ premises $\pi_0=\,\Delta'\vdash P[V/x]:\bigparr_{i=1}^k\bigotimes_{j=1}^{n_i}{(\tau_{ij}\multimap\alpha_{ij})}$ and $\pi_i=\,\Gamma_i\vdash Q[V/x]:\bigparr_{j=1}^{n_i}\tau_{ij}$ for $i=1,\dots, k$, with $\Delta=\Delta'\otimes\bigotimes_{i=1}^k\Gamma_i$, $\alpha=\bigparr_{i=1}^k\bigparr_{j=1}^{n_i}\alpha_{ij}$ and $\mes{\pi}=\sum_{i=0}^k\pi_i+(\sum_{i=1}^k2n_i)-1$.
		Then, by the induction hypothesis, we get $\pi_0'=\,\Delta'\vdash(\lambda x.P)V:\bigparr_{i=1}^k\bigotimes_{j=1}^{n_i}(\tau_{ij}\multimap\alpha_{ij})$, and
		$\pi_i'=\,\Gamma_i\vdash(\lambda x.Q)V:\bigparr_{j=1}^{n_i}\tau_{ij}$, with $\mes{\pi_i'}=\mes{\pi_i}+1$.
		Hence by rule $\multimap_E$ we obtain $\pi''=\,\Delta'\otimes\bigotimes_{i=1}^k\Gamma_i\vdash((\lambda x.P)V)((\lambda x.Q)V):\bigparr_{i=1}^k\bigparr_{j=1}^{n_i}\alpha_{ij}$, with $\mes{\pi''}=\sum_{i=0}^k\mes{\pi_i'}+(\sum_{i=1}^k2n_i)-1$. By the above claim, we get $\pi'=\,\Delta'\otimes\bigotimes_{i=1}^k\Gamma_i\vdash(\lambda x.PQ)V:\bigparr_{i=1}^k\bigparr_{j=1}^{n_i}\alpha_{ij}$ such that $\mes{\pi'}=\mes{\pi''}-k=\mes{\pi}+1$. \qed

\end{proof}

\subsection{Convergence}\label{subsec:Convergence}

%In this section we finally provide the main result of the paper.
From our ``quantitative'' versions of subject reduction and subject expansion one easily obtains that
our type system captures exactly the weakly normalising terms, and that the size $\mes{\pi}$ of a derivation tree $\pi =\ \vdash M : \alpha$ decreases along the reduction of $M$.
However, when $\alpha$ satisfies in addition a suitable minimality condition (namely the fact that $\alpha $ is of shape $ \one\parr\cdots\parr\one$), then we can be more precise and say that there exists a reduction from $M$ to a normal form,
having length \emph{exactly} $\mes{\pi}$.

In the following $\parr^k \one$, with $k > 0$,  stands for $\one\parr\cdots\parr\one$ ($k$ times).

\begin{theorem}\label{thm:Convergence}
Let $M$ be a closed term, and $k> 0$. 
There is a typing tree $\pi$ for $\vdash M:\parr^k\one$ iff there are values $V_1, \dots, V_k$ and a reduction $M\msto V_1\parallel\cdots\parallel V_k$ of length $\mes{\pi}$.
%\todom[inline]{do you have a better logical equivalent statement? also the proof is a little bit redundant}
%\todoa[inline]{My proposal:
%
%$M$ is typable if, and only if, $M$ converges.
%
%In addition, $\vdash M:\bigparr_{i=1}^k\one$ if, and only if, $M\msto V_1\parallel\cdots\parallel V_k$ in $\mes{\pi}$ steps}
%\todog[inline]{We need the existential quantifiers on $V$'s and the reduction}
\end{theorem}
\begin{proof}~
$(\Rightarrow)$ Suppose $\pi=\ \vdash M: \parr^k \one$.
	We proceed by induction on $\mes{\pi}$.  If $M = V_1\parallel\cdots\parallel V_{k'}$, then $\pi$ must start with a tree of $k'-1$ rules $\parallel$, and then $k'$ rules $\multimap_I$ with conclusion, respectively, $\ \vdash V_1:\one,\dots, \ \vdash V_{k'}:\one$. We then have $k=k'$, and $M$ trivially converges to $V_1\parallel\cdots\parallel V_{k'}$ in $\mes{\pi}=0$ steps.

Otherwise, since $M$ is closed, there exists $N$ such that $M\to N$.
By Theorem~\ref{thm:SR}, %there exists $N$ such that $M\to N$ and $\pi'=\ \vdash N:\parr^k \one$, with $\mes{\pi'} = \mes{\pi}-1$.
such an $N$ can be chosen in such a way $\pi'=\ \vdash N:\parr^k \one$, with $\mes{\pi'} = \mes{\pi}-1$.
From the induction hypothesis we know that $N$ converges in $\mes{\pi'}$ steps to $V_1\parallel\cdots\parallel V_k$. Therefore, $M$ converges in $\mes{\pi'}+1=\mes{\pi}$ steps to $V_1\parallel\cdots\parallel V_k$.

$(\Leftarrow)$ Suppose that $M\msto V_1\parallel\cdots\parallel V_k$.
By Remark~\ref{rmk:valuesTop}, there is $\pi=\ \vdash V_1\parallel\dots\parallel V_k : \parr^k \one$ and $\mes{\pi}=0$.
Therefore, by the subject expansion (Theorem~\ref{thm:SE}) there is $\pi'=\ \vdash M:  \parr^k \one$ and $\mes{\pi'}$ is equal to the length of the reduction $M\msto V_1\parallel\cdots\parallel V_k$. \qed
% If $M$ is a value, then by Remark~\ref{rmk:valuesTop} it is typable, and it is trivially convergent. If $M$ is not a value:
% \begin{description}
% \item[$(\Rightarrow)$] Let $\pi=\vdash M:\alpha$. We proceed by induction on $|\pi|$. Since $M$ is not a value, $|\pi|>[1]$ and since it is a closed term, it must reduce. Then by Theorem~\ref{thm:SR}, there exists $N$ such that $M\to N$ and $\pi'= \vdash N:\alpha$, with $|\pi'|<|\pi|$. By the induction hypothesis $N$ is convergent, and so $M$ is.
%
% \item[$(\Leftarrow)$] Let $M$ be convergent, then its normal form is typable since a closed normal term is a value, and by Remark~\ref{rmk:valuesTop}, every value is typable. Thus by Theorem~\ref{thm:SE}, $M$ has a type.
%As for the equivalence ``typable iff convergent'', the left-to-right implication is an easier variant of the previous case $(\Rightarrow)$. The converse is a trivial consequence of the case $(\Leftarrow)$.
\end{proof}

%From an analogous (actually simpler) proof one obtains the following result.

\begin{corollary} Let $M$ be closed, then $M$ is typable if and only if $M$ converges.
\end{corollary}
\section{Adequacy and (Lack of) Full Abstraction}\label{sec:NoFA}

The choice of presenting a model through a type discipline or a reflexive object is
more a matter of taste rather than a technical decision. (Compare for instance the type system of \cite{PaganiR10} and the interpretation of \cite{BucciarelliEM12}).
The model $\mathcal{V}$ associated with our type system lives in the category {\bf Rel} of sets and relations (refer to \cite{Ehrhard12} for more details) and is defined by $\mathcal{V} = \bigcup_{n\in\nat}\mathcal{V}_n$, with
\[
	\mathcal{V}_0 = \emptyset,\qquad \mathcal{V}_{n+1} = \Mfin{\mathcal{V}_n}\times\Mfin{\Mfin{\mathcal{V}_n}},
\]
where $\Mfin{X}$ denotes the set of finite multisets over a set $X$. % and $\times$ is the cartesian product of two sets. 
In fact, $\Mfin{X}$ interprets in {\bf Rel} the exponentials $\oc X$ and $\wn X$, whilst the cartesian product is the linear implication $\multimap$, so that $\mathcal{V}$ is the minimal solution of the equation $\mathcal{V}\simeq\oc \mathcal{V} \multimap\wn\oc \mathcal{V}$. Recalling Equation \ref{eq:intro_types} in the introduction, this means that the object $\mathcal{V}$ represents ``value types'',  while computational types $\LTypes$ will be represented by elements of $\mathcal{C} = \oc \mathcal{V}= \Mfin{\mathcal{V}}$ and parallel-types $\Types$ as elements of $\mathcal{T} = \wn \mathcal{C} = \Mfin{\mathcal{C}}$.
This intuition can be formalized by defining two injections $(\cdot)^\circ : \Types\to\mathcal{T}$ and $(\cdot)^\bullet : \LTypes\to\mathcal{C}$ by mutual induction,
as follows: $\tau^\circ = [\tau^\bullet]$, $(\alpha\parr \beta)^\circ = \alpha^\circ\mcup \beta^\circ$, $\one^\bullet = []$, $(\tau\otimes\rho)^\bullet = \tau^\bullet\mcup\rho^\bullet$ and $(\tau\multimap\alpha)^\bullet = [(\tau^\bullet,\alpha^\circ)]$. 

It is beyond the scope of the present paper to give the explicit inductive definition of the interpretation of terms. % in $\mathcal{V}$.
For our purpose it is enough to know that such an interpretation can be characterised (up to isomorphism) as follows.

\begin{definition} The \emph{interpretation of a closed term $M$} is defined by $\Int{M} = \{ \alpha \mid \ \vdash M : \alpha \}\subseteq\Types$.
\end{definition}

%When $M$ is closed then $\Int{M}$ will be considered as a plain subset of $\Types$.

The interpretations of terms are naturally ordered by set-theoretical inclusion;
an interesting problem is to determine whether there is a relationship between
this ordering and the following observational preorder on terms.

\begin{definition}[Observational preorder] Let $M, N\in\Terms$ be closed.
We set $M \sqle N$ iff for all closed terms $\vec P$, $M\vec P$ converges implies that $N\vec P$ converges.
\end{definition}

A model is called \emph{adequate} if $\Int{M}\subseteq \Int{N}$ entails $M\sqle N$;
it is called \emph{fully abstract} if in addition the converse holds.

The adequacy of the model $\mathcal{V}$ follows easily from Theorem~\ref{thm:Convergence}
and the monotonicity of the interpretation.

\begin{corollary}[Adequacy] For all $M,N$ closed, if $\Int{M}\subseteq \Int{N}$ then $M\sqle N$.
\end{corollary}

On the contrary, $\mathcal{V}$ is not fully abstract.
This is due to the fact that the call-by-value \lam-calculus admits the creation of an `ogre' that is able to `eat' any finite sequence of arguments and converge,
constituting then a top of the call-by-value observational preorder. 
Following \cite{BoudolIC94}, we define the ogre as 
$
	\ogre = \halfogre\halfogre$ where $\halfogre =  \lam xy.xx.
$
The ogre $\ogre$ converges since $\ogre\to\lam y.\ogre$ and remains convergent when applied to every sequence of %(parallel compositions of) 
values,
by discarding them one at time.

\begin{lemma} For all closed terms $M$ we have $M\sqle \ogre$.
\end{lemma}
\begin{proof}
Given a term $M$ and a sequence $\vec P = P_1\cdots P_k$ of closed terms it is easy to check that $M\vec P$ can converge only when $\vec P$ converges.
In that case we have $\ogre \vec P\msto (\lam y.\ogre) (V_1\parallel\cdots\parallel V_n)P_2\cdots P_k\msto \ogre P_2\cdots P_k\parallel\cdots\parallel\ogre P_2\cdots P_k   \msto \lam y.\ogre\parallel\cdots\parallel \lam y.\ogre$.
Therefore $\ogre$ is maximal with respect to $\sqle$.\qed
\end{proof}

It is easy to check that $\one$ and $(\one\multimap\one)\otimes(\one\multimap(\one\multimap\one))$ are valid types for $\ogre$, and thus belong to its interpretation.
The following lemma gives a precise characterisation of $\Int{\ogre}$.

\begin{lemma} $\alpha\in\Int{\ogre}$ iff $\alpha=\bigotimes_{i=0}^n(\one\multimap\alpha_i)$ with $n\geq 0$ and $\alpha_i\in\Int{\ogre}$ for all $i\leq n$. 
In particular, we have that $\Int{\mathbf{I}}\not\subseteq\Int{\ogre}$.
% = \{\tau \mid \exists k,n_1,\dots,n_k\geq 0,\ \tau=\bigotimes^{n_1}(\one\multimap\cdots\bigotimes^{n_k}(\one\multimap\one)) \}$.
\end{lemma}
\begin{proof} The crucial point is to remark that $\ogre\to\lambda y.\ogre$, so by Theorem~\ref{thm:SR}~and~\ref{thm:SE}, we get $\Int{\ogre}=\Int{\lambda y.\ogre}$. 
Therefore we have the following chain of equivalences:
\begin{align*}
&\alpha\in\Int{\ogre}\text{ iff }\alpha\in\Int{\lambda y.\ogre}\\
&\text{ iff }\alpha=\otimes_{i=0}^n(\tau_i\multimap\alpha_i)\in\Int{\lambda y.\ogre}&\text{by Lemma~\ref{lem:values}, $n\geq0$}\\
&\text{ iff }\alpha=\otimes_{i=0}^n(\tau_i\multimap\alpha_i)\text{ and }\forall i, \tau_i\multimap\alpha_i\in\Int{\lambda y.\ogre}&\text{by Lemma~\ref{lem:intersec}}\\
&\text{ iff }\alpha=\otimes_{i=0}^n(\tau_i\multimap\alpha_i)\text{ and }\forall i, \tau_i=\one \text{ and }\alpha_i\in\Int{\ogre}&\text{since $y\notin\FV(\ogre)$.}
\end{align*}
We have that $\Int{\mathbf{I}}\not\subseteq\Int{\ogre}$ as, for instance, ${(\one\multimap\one)\multimap(\one\multimap\one)}\in \Int{\mathbf{I}}\setminus\Int{\ogre}$.\qed
%
%
% is a regular term, the type system looses the $\parr$-layers and we can focus on \secondsort s.
%All the type derivations for $\halfogre$ are of the following form:
%  $$\prooftree
%	\prooftree
%	  \prooftree
%		\prooftree
%		\justifies x:\tau_{ij}\multimap\rho_{ij}\vdash x:\tau_{ij}\multimap\rho_{ij}
%		\using ax
%		\endprooftree
%		\prooftree
%		\justifies x:\tau_{ij}\vdash x:\tau_{ij}
%		\using ax
%		\endprooftree
%	  \justifies x:(\tau_{ij}\multimap\rho_{ij})\otimes\tau_{ij}\vdash xx:\rho_{ij}
%	  \using\multimap_E
%	  \endprooftree
%	  \hspace{-0.8cm}\textrm{\scriptsize$j=1,\dots,m_i$}
%	\justifies x:\bigotimes_{j=1}^{m_i}((\tau_{ij}\multimap\rho_{ij})\otimes\tau_{ij})\vdash\lambda y.xx:\bigotimes_{j=1}^{m_i}(\one\multimap\rho_{ij})
%	\using\multimap_I
%	\endprooftree
%	\hspace{-0.5cm}\textrm{\scriptsize$i=1,\dots,n$}
%  \justifies\vdash\lambda xy.xx:\bigotimes_{i=1}^n((\bigotimes_{j=1}^{m_i}((\tau_{ij}\multimap\rho_{ij})\otimes\tau_{ij}))\multimap\bigotimes_{j=1}^{m_i}(\one\multimap\rho_{ij}))\qquad(\star)
%  \using\multimap_I
%  \endprooftree$$
%In order to type $\ogre$ we need to derive (1) $\vdash\halfogre : \tau\multimap\rho$ and (2) $\vdash\halfogre : \tau$ for some $\tau\in\LTypes$.
%Since (1) has an arrow type, in $(\star)$ we need to take $n=1$, so we get 
%$\tau = \bigotimes_{j=1}^{m}((\tau_{j}\multimap\rho_{j})\otimes\tau_{j}))$ and $\rho = \bigotimes_{j=1}^{m}(\one\multimap\rho_{j})$
%\todog[inline]{I need to conclude }
\end{proof}

Summing up, get that $\mathbf{I}\sqle \ogre$, while $\Int{\mathbf{I}}\not\subseteq\Int{\ogre}$.
%Our counterexample also applies to the system presented in \cite{Ehrhard12}, while the (non resource sensitive) model of \cite{BoudolIC94} is fully abstract for a call-by-value \lam-calculus with non-deterministic choice.
%It is an interesting open problem to find a relational model fully abstract for the call-by-value \lam-calculus.

% !TEX root = ../main.tex
%
\section{Conclusion and future work}

We introduced a call-by-value non-deterministic $\lambda$-calculus with a type system ensuring convergence.
We proved that such a type system gives a bound on the length of the lazy call-by-value reduction sequences, which is the exact length when the typing is minimal.
Finally, we show that the relational model $\mathcal{V}$ capturing our type system is adequate, but not fully abstract.

As our counterexample to full abstraction contains no non-deterministic operators, it also holds for the standard call-by-value \lam-calculus and the relational model described in~\cite{Ehrhard12}. This is a notable difference with the call-by-name case, where the relational model is proven to be fully abstract for the pure call-by-name \lam-calculus  \cite{Manzonetto09}, while other counterexamples (see~\cite{BucciarelliEM12,Breuvart}) break full abstraction in presence of may or must non-deterministic operators. 
An open problem is to find a relational model fully abstract for the call-by-value \lam-calculus.

Various fully abstract models of may and must non-determinism are known in the setting of Scott domain based semantics and idempotent intersection types. 
In particular, for the call-by-value case we mention \cite{BoudolIC94,DezaniLP98}. Comparing these models and type systems with the ones issued from the relational semantics is a research direction started in \cite{Ehrhard12} with some notable results. It would be interesting to reach a better understanding of the role played by intersection idempotency in the question of full abstraction.  

%% Relation with algebraic calculi and future work in that direction
%The choice of endowing the \lam-calculus with two different operators modelling may and must non-determinism is widespread in the literature \cite{BucciarelliEM12,DezanidP96,BoudolIC94,DezaniLP98}.
%Another possible presentation is to consider a unique non-deterministic operator, written say $M+N$, and
%to represent may and must non-determinism as two possible notions of convergence. % (converging to a value for at least one choice, or converging for all choices).
%%For instance this approach is followed in \cite{PaganiR10} to analyze the may and must semantics of Tranquilli's resource \lam-calculus.
Another axis of research is to generalize our approach to study the convergence in (call-by-name and call-by-value) $\lambda$-calculi with richer algebraic structures than simply may/must non-deterministic operators, such as \cite{VauxMSCS09,ArrighiDowekRTA08}. % \cite{ArrighiDiazcaroLMCS12,ArrighiDiazcaroValironDCM11}.
%of the (call-by-name) algebraic $\lambda$-calculus~\cite{VauxMSCS09} and of the (call-by-value) linear-algebraic $\lambda$-calculus (Lineal)~\cite{ArrighiDowekRTA08}.
In these calculi the choice operator is enriched with a weight, i.e.\ sums of terms are of the form $\alpha.M+\beta.N$, where $\alpha,\beta$ are scalars from a given semiring, pondering the choice. We would like to design type systems characterizing convergence properties in these systems. First steps have been done in \cite{ArrighiDiazcaroLMCS12,ArrighiDiazcaroValironDCM11}.

\bigskip
{\bf Acknowledgements.}
We wish to thank Thomas Ehrhard and Simona Ronchi Della Rocca for interesting discussions,
and the anonymous reviewers for their careful reading.

\bibliographystyle{splncs_srt}
\bibliography{include/biblio}

\begin{thebibliography}{10}

\bibitem{AmadioC98}
Amadio, R., Curien, P.L.:
\newblock Domains and Lambda-Calculi.
\newblock Number~46 in Cambridge Tracts in Theoretical Computer Science.
  Cambridge University Press (1998)

\bibitem{ArrighiDiazcaroLMCS12}
Arrighi, P., D{\'\i}az-Caro, A.:
\newblock A {S}ystem {F} accounting for scalars.
\newblock Logical Methods in Computer Science \textbf{8}(1:11) (2012)

\bibitem{ArrighiDiazcaroValironDCM11}
Arrighi, P., D{\'\i}az-Caro, A., Valiron, B.:
\newblock A type system for the vectorial aspects of the linear-algebraic
  $\lambda$-calculus.
\newblock In: DCM'11. Volume~88 of EPTCS. (2012)  1--15

\bibitem{ArrighiDowekRTA08}
Arrighi, P., Dowek, G.:
\newblock Linear-algebraic lambda-calculus: higher-order, encodings, and
  confluence.
\newblock In: RTA'08. Volume 5117 of LNCS., Springer (2008)  17--31

\bibitem{Bare}
Barendregt, H.:
\newblock The lambda calculus: its syntax and semantics.
\newblock North-Holland, Amsterdam (1984)

\bibitem{BernadetL11}
Bernadet, A., Lengrand, S.:
\newblock Complexity of strongly normalising $\lambda$-terms via non-idempotent
  intersection types.
\newblock In: FOSSACS 2011. (2011)  88--107

\bibitem{BoudolIC94}
Boudol, G.:
\newblock Lambda-calculi for (strict) parallel functions.
\newblock Information and Computation \textbf{108}(1) (1994)  51--127

\bibitem{Breuvart}
Breuvart, F.:
\newblock On the discriminating power of tests in the resource
  $\lambda$-calculus Submitted. Draft available at
  http://hal.archives-ouvertes.fr/hal-00698609.

\bibitem{BucciarelliEM12}
Bucciarelli, A., Ehrhard, T., Manzonetto, G.:
\newblock A relational semantics for parallelism and non-determinism in a
  functional setting.
\newblock APAL \textbf{163}(7) (2012)  918--934

\bibitem{CoppoDezani78}
Coppo, M., Dezani-Ciancaglini, M.:
\newblock A new type-assignment for $\lambda$-terms.
\newblock Archiv f{\"u}r Math. Logik \textbf{19} (1978)  139--156

\bibitem{deCarvalho}
{de Carvalho}, D.:
\newblock Execution time of lambda-terms via denotational semantics and
  intersection types.
\newblock To appear in \emph{Math. Struct. in Comp. Sci.} (2008)

\bibitem{DezanidP96}
Dezani-Ciancaglini, M., de'Liguoro, U., Piperno, A.:
\newblock Filter models for conjunctive-disjunctive lambda-calculi.
\newblock Theor. Comp. Sci. \textbf{170}(1-2) (1996)  83--128

\bibitem{DezaniLP98}
Dezani-Ciancaglini, M., de'Liguoro, U., Piperno, A.:
\newblock A filter model for concurrent lambda-calculus.
\newblock SIAM J. Comput. \textbf{27}(5) (1998)  1376--1419

\bibitem{Ehrhard12}
Ehrhard, T.:
\newblock Collapsing non-idempotent intersection types.
\newblock In: CSL'12. Volume~16 of LIPIcs. (2012)  259--273

\bibitem{Girard87}
Girard, J.Y.:
\newblock Linear logic.
\newblock Theoretical Computer Science \textbf{50} (1987)  1--102

\bibitem{Krivine90}
Krivine, J.L.:
\newblock Lambda-calcul: types et mod{\`e}les.
\newblock {\'E}tudes et recherches en informatique. Masson (1990)

\bibitem{LaurentPhD}
Laurent, O.:
\newblock \'Etude de la polarisation en logique.
\newblock PhD thesis, Universit\'e de {A}ix-{M}arseille II, France (2002)

\bibitem{Manzonetto09}
Manzonetto, G.:
\newblock A general class of models of $\mathcal{H}^{\star}$.
\newblock In: MFCS'09. Volume 5734 of LNCS., Springer (2009)  574--586

\bibitem{MaraistOTW99}
Maraist, J., Odersky, M., Turner, D.N., Wadler, P.:
\newblock Call-by-name, call-by-value, call-by-need and the linear
  $\lambda$-calculus.
\newblock Theor. Comp. Sci. \textbf{228}(1-2) (1999)  175--210

\bibitem{PaganiR10}
Pagani, M., Ronchi Della~Rocca, S.:
\newblock Linearity, non-determinism and solvability.
\newblock Fundam. Inform. \textbf{103}(1-4) (2010)  173--202

\bibitem{PlotkinTCS75}
Plotkin, G.D.:
\newblock Call-by-name, call-by-value and the $\lambda$-calculus.
\newblock Theor. Comp. Sci. \textbf{1}(2) (1975)  125--159

\bibitem{Salle80}
Sall\'e, P.:
\newblock Une g\'en\'eralisation de la th\'eorie de types en $\lambda$-calcul.
\newblock RAIRO: {I}nformatique {T}h\'eorique \textbf{14}(2) (1980)  143--167

\bibitem{VauxMSCS09}
Vaux, L.:
\newblock The algebraic lambda calculus.
\newblock Math. Struct. in Comp. Sci. \textbf{19}(5) (2009)  1029--1059

\end{thebibliography}
\end{document}